\documentclass[runningheads,a4paper]{llncs}

\usepackage{amsmath}
\usepackage{amssymb}
\usepackage{graphicx}
\usepackage{color}
\usepackage{paralist}

\usepackage{url}
\urldef{\mailsa}\path|{alfred.hofmann, ursula.barth, ingrid.haas, frank.holzwarth,|
\urldef{\mailsb}\path|anna.kramer, leonie.kunz, christine.reiss, nicole.sator,|
\urldef{\mailsc}\path|erika.siebert-cole, peter.strasser, lncs}@springer.com|

\begin{document}

\mainmatter  

\title{Retargeting Without Tracking}


\author{Minh-Dung Tran\and Gergely Acs\and Claude Castelluccia}
%

\institute{INRIA, France\\
\path|{minh-dung.tran, gergely.acs, claude.castelluccia}@inria.fr|\\
\url{http://www.inria.fr}}


\maketitle

\begin{abstract}
Retargeting ads are increasingly prevalent on the Internet as their effectiveness has been shown to outperform conventional targeted ads. Retargeting ads are not only based on users' interests, but also on their intents, i.e. commercial products users have shown interest in. Existing retargeting systems heavily rely on tracking, as retargeting companies need to know not only the websites a user has visited but also the exact products on these sites. They are therefore very intrusive, and privacy threatening. Furthermore, these schemes are still sub-optimal since tracking is partial, and they often deliver ads that are obsolete (because, for example, the targeted user has already bought the advertised product). 
\\\\
This paper presents the first privacy-preserving retargeting ads system. In the proposed scheme, the retargeting algorithm is distributed between the user and the advertiser such that no systematic tracking is necessary, more control and transparency is provided to users, but still a lot of targeting flexibility is provided to advertisers. We show that our scheme, that relies on homomorphic encryption, can be efficiently implemented and trivially solves many problems of existing schemes, such as frequency capping and ads freshness.
\end{abstract}

\section{Introduction}
\label{sec:intro}
In targeted advertising, companies track user online browsing activities to infer user information, such as age, gender and interests, and then personalize ads. Targeting helps advertisers to optimally allocate their advertising resources to their most likely potential customers, and thus to increase their revenue. As a result, companies have been constantly improving their tracking and ad personalizing technologies with the aim to enhance targeting performance. 

\textit{Retargeting} ads have been introduced in recent years with the aim to match the exact user attention or previous online action. For example, a user who has visited hotels.com looking for a hotel in Paris will very likely receive frequent ads about this hotel during his subsequent browsing sessions, for instance on accuweather.com. Advertisers (hotels.com in this case) aim to bring these customers back to their sites by showing ads related to the products they previously showed interest in. Retargeting ads have been shown to be significantly effective; Criteo in particular confirmed that personalized retargeting ads perform 6 times better than general ads \cite{criteo6x}. Increasingly, retargeting is becoming prevalent in travel, real estate and financial services industries \cite{nytimes_retargeting}.

Retargeting advertisers, mostly commercial online stores (e.g., hotels.com), often leverage a third party, called \textit{retargeter} (e.g., Criteo), to handle the retargeting task. Retargeters track users on these stores to collect products that they are interested in, and then select one to advertise to a user when (s)he visits an ad-enabled website. This is beneficial to advertisers as they can outsource the optimization of the whole advertising process  to an external party with dedicated resource and expertise. However, as users' interested products are centrally collected by third-party retargeters, this also poses significant privacy threats. For example, these products are shown to reveal users' important events in their life, such as being pregnant, getting divorced or graduating \cite{shopping_habits}. Compared to conventional tracking, where ad networks usually only collect urls of sites visited by a user (e.g., hotels.com) to infer his interest categories (e.g., traveling), retargeting trackers also retrieve \emph{exact products} on each page thereby inferring more accurate information about the user (e.g., the city where he is searching for a hotel).

There have been serious public concerns about the prevalence and resulting privacy threats of retargeting. As observed by the author of a The New York Times's article \cite{nytimes_retargeting}: ``Retargeting has reached a level of precision that is leaving consumers with the palpable feeling that they are being watched as they roam the virtual aisles of online stores." and ``It illustrates that there is a commercial surveillance system in place online that is sweeping in scope and raises privacy and civil liberties issues". 

The natural reaction from the user community is to block trackers. There are a bundle of tools for this purpose, such as AdBlockPlus \cite{adblockplus}, Ghostery \cite{ghostery}, DoNotTrackMe \cite{dntme}, TrackMeNot \cite{trackmenot} and much more. In addition, privacy advocates proposed Do-Not-Track \cite{dnt} initiative with the aim to help users notify trackers of their non-tracking preference and to urge trackers to stop tracking when receiving such notification. Unfortunately, these anti-tracking approaches prevent targeted advertising and compromise the business model of the Internet, which is mainly fostered by advertising revenue. The prevalence of these initiatives potentially leads to a situation where no one benefits. For example, advertisers cannot effectively promote their products and consequently lose sale revenue, content providers lose advertising revenue and may stop providing free content, and finally users do not receive useful ads and may have to pay for the access to content which is currently free. 

Alternatively, several research proposals \cite{privad}\cite{adnostic} aimed to shift advertisers' ad personalizing algorithms to users' devices, thus providing users with complete control over their data. However, it is unclear in these approaches  how the confidentiality of these algorithms is guaranteed against users. In addition, these proposals do not consider the \textit{Real-Time Bidding} (RTB) protocol, which allows trading ad spaces at real-time auctions on a per-ad-impression basis, in their design. RTB is actually a major channel for retargeters to buy ad spaces.

\paragraph{Contribution:} We propose the first retargeting system which does not require tracking. In this system, a client software operating at the user's terminal collects products on websites that the user has visited, and stores them locally. The product selection algorithm is distributed between the client and the retargeter, while an effective homomorphic scheme protects the algorithm's confidentiality from the client. Specifically, the homomorphic encryption scheme allows the client to perform some precomputation for the retargeter without learning how exactly the retargeter selects the advertised products; these precomputation results are then sent to the retargeter in a way that does not leak any user private information (e.g., IP address); and finally, the retargeter selects the final products to advertise to the user. Our scheme is compatible with today's advertising systems, and integrates RTB as part of its design.

The proposed scheme has several benefits over the existing retargeting scheme. 
(1) It improves {\em user privacy} by preventing systematic tracking.
(2) It provides {\em more transparency and control} for users over their data which is used for retargeting. In particular, 
users can filter out privacy-sensitive products from advertising.
(3) It is {\em more efficient} than existing schemes since profiling is performed locally, and is therefore based on higher quality data. 
Furthermore, in existing systems, users frequently receive retargeting ads about products that are not relevant anymore (e.g., they already bought them from another seller). This is annoying to most users, yet inevitable, since advertisers do not always know whether users have changed their intents or made a purchase.  
This problem is trivially solved in our scheme since users can filter-out products that they are no longer interested in. 
(4) It does not rely on cookie and tracking technologies, and therefore works even if users use anti-tracking tools.

\paragraph{Organization:} The paper is organized as follows. 
Section \ref{sec:background} presents an overview of the current retargeting system and related privacy risks. We then describe the goals and security assumptions of the proposed scheme in Section \ref{sec:goal}. We give an overview of the scheme in Section 
\ref{sec:overview}, and clarify the details in Section \ref{sec:details}. A privacy analysis is presented in Section \ref{sec:analysis}. The implementation and performance evaluation are presented in Section \ref{sec:exp}. We discuss possible improvements in Section \ref{sec:discuss}, survey related work in Section \ref{sec:related} and conclude in Section \ref{sec:conclusion}. 
\section{Background}
\label{sec:background}
Before retargeting became prevalent, conventional targeted ad systems (e.g., Google Adsense) were often only interest-based: ad networks collected urls of sites visited by users (e.g., hotels.com) to infer user interest categories (e.g., traveling) and used this information to deliver personalized ads. 

By contrast, \emph{retargeting} is much more effective (and also privacy-invasive): retargeting trackers do not only attempt to identify user interests, but also aim to get the \emph{exact products} on each page, possibly with additional information such as related user actions (e.g., search) and the point where the user suspended the purchasing process. For instance, given a user visiting hotels.com, retargeters might learn that he is interested in a hotel in a certain district in Paris. In addition, they can also infer whether the user only looks at this hotel, or has a clear booking intent (the hotel is in a shopping cart), or already made a booking (the hotel is in a booking confirmation page).  

\subsection{Retargeting Mechanism}

There are five entities in a retargeting system: \textit{advertiser}, \textit{publisher}, \textit{ad exchange}, \textit{retargeter} and \textit{user}. Advertisers (e.g., hotels.fr) wish to promote their products by showing ads to users. Publishers (e.g., nytimes.com) develop web pages providing content to users and sell ad spaces on their pages to advertisers. Ad Exchanges (e.g., DoubleClick) connect ad buyers (advertisers or their representatives) and sellers (publishers) in real-time transactions. Retargeters (e.g., Criteo) provide retargeting service to advertisers. 

We illustrate how retargeting works with the following example. A user visits maty.com, searches for an engagement ring (\textit{product}), looks at its details and then leaves the website without a purchase. Later on, he visits accuweather.com and finds this engagement ring advertised to him on the page. In this scenario, the owner of maty.com (\textit{advertiser}) uses the retargeting service provided by a \textit{retargeter}, say Criteo, to retarget the user on \textit{publisher} pages (accuweather.com in this example). This retargeting process has two main phases: tracking (the retargeter tracks users on maty.com) and delivering ads (the retargeter delivers ads to users on accuweather.com). In what follows, we elaborate these two phases.

\paragraph{Tracking:} The advertiser puts the retargeter's tracking code on its pages and encodes in each page the ids of products aimed for retargeting (e.g., an engagement ring's id: ``ring123"). When a user visits such a page, the tracking code sends his cookie (belonging to the retargeter's domain) along with these product ids to the retargeter. In addition, the advertiser can flexibly send additional user information to the retargeter. For example, if the user performs a search, the keyword and product ids on the resulting page could be sent to the retargeter. In case the user visits his shopping cart, the quantities of products in the cart could also be sent along with the product ids. This information serves a more accurate targeting (e.g., targeting shopping cart abandoners). As a result of tracking, the retargeter builds a list of potentially advertised products assigned to each user cookie.

\paragraph{Delivering ads:} The retargeter finds ad spaces on publisher pages (e.g., accuweather.com) mostly through the \emph{Real-Time Bidding} (RTB) protocol. RTB is provided by an Ad Exchange (ADX): publishers put the ADX's advertising code on their web pages; when a user visits such a page, the code sends an ad request, which contains the ADX's user cookie and the information about the page, to the ADX. The ADX subsequently broadcasts these data in form of bid requests to its registered bidders, including retargeters, for them to compete in a real-time auction for the ad space. 

Upon receiving a bid request, each bidder recognizes its own user cookie from the ADX's cookie thanks to a cookie matching protocol \cite{cm}. Given the list of products previously assigned to the user cookie, each retargeter selects some products\footnote{The retargeter can use the exact products that the user visited and/or suggest relevant products.} to be advertised to the user as well as the advertising price it would pay. This selection is typically based on the products, the user profile, and the page which displays ads.

ADX is often operated by giant firms, such as Google, Yahoo or Microsoft: they manage huge online ad inventories and put them into real-time auctions in order to maximize the revenue.

\subsection{Privacy Concerns of Retargeting}
\label{sec:privacy}

There are serious privacy concerns related to the fact that products, which users are interested in, are centrally collected by retargeters. Specifically, certain kinds of products might immediately expose sensitive information: an engagement ring reveals that the user likely intends to get married while its price implies his financial capacity, a hotel booking reveals the user's destination, a bank loan reveals the user's financial status, and so forth. In addition, research has shown that users' habits can be characterized by the products they consume for a period, and that changing habits anticipate special events in their life (e.g., being pregnant, getting divorced or graduating) \cite{shopping_habits}. The inferred information can be used by marketers to better suggest products to users, but can also enable price or service discrimination \cite{pricediscrimination}. Moreover, retargeting ads might reveal users' private actions, e.g., to their family members. This has been shown in a practical case when a commercial coupon advertised to a father revealed  that his teenage daughter is pregnant \cite{shopping_habits}. 
Furthermore, all this private information can be easily linked with user identity, as commercial websites often incentivize users to provide their name, email address or telephone number, e.g., through a fidelity program. 
\section{Goals and Assumptions}
\label{sec:goal}

\subsection{Goals}
We introduce a novel retargeting system which preserves \textit{user privacy} from retargeters. Specifically, our scheme ensures that retargeters cannot associate any user information with user personal identity such as IP address (\textit{anonymity}) and cannot associate multiple pieces of such information with the same user (\textit{unlinkability}). 
We also prevent any man-in-the-middle attack (possibly mounted by Ad Exchanges) between users and retargeters that could aim at eavesdropping the transmitted data (e.g., ads) to learn private information (e.g., list of products visited by the user). 

In addition, our scheme also protects the secrecy of retargeters which do not want to reveal every detail of their ad selection algorithm even if it is distributed between the user and the retargeter. This is a challenging task since the algorithm needs  private user attributes (e.g., age, sex, or interest categories) and also confidential data from the retargeter (e.g., the  combination of user attributes that yields higher clicking rate) as input. 

While ensuring privacy, we also aim to keep the retargeting \textit{effectiveness} of today's systems. Specifically, we provide retargeters with almost the same input and flexibility for ad selection as they have today. Nevertheless, we cannot compare the performance of our proposal with that of the current retargeting system due to the lack of real data and details of  today's ad selection algorithms. Hence, we leave this performance comparison for future work, and rather discuss several possible improvements of our scheme. 


\subsection{Security Assumptions}
\label{security}
Some related work \cite{privad}\cite{adnostic} assume honest-but-curious parties, which abide by the protocol rules but may misuse any information obtained in the protocol run. In our work, we assume that the retargeter and the ad exchange can be \emph{active}, i.e., they can also mount active attacks to break anonymity or unlinkability. 

However, we assume that the ad exchange does not collude with the retargeter. The major goal of an ad exchange is to provide a fair market for trading ads, not to break user privacy at all cost. Moreover, ad exchanges are often large companies with reputation (e.g., Google or Facebook), which unlikely take the risk to collude with an external party. In addition, most of their privacy policies actually contain non-collusion terms making them subject to legal action. Entities which do not have such privacy statement could be excluded by the client software by maintaining a blacklist of them.

On the client side, we assume that the user trusts the client software. The client software can be open-source (e.g., a browser plugin), and therefore can be easily audited by a trusted party. Finally, the client can be malicious towards retargeters. For example, a retargeter's competitor might manipulate the local configuration to learn the private algorithm of the retargeter. 

\section{System Overview}
\label{sec:overview}

\begin{figure}[t!]
\label{fig:sys_overview}
  \centering
  \includegraphics[width=\textwidth]{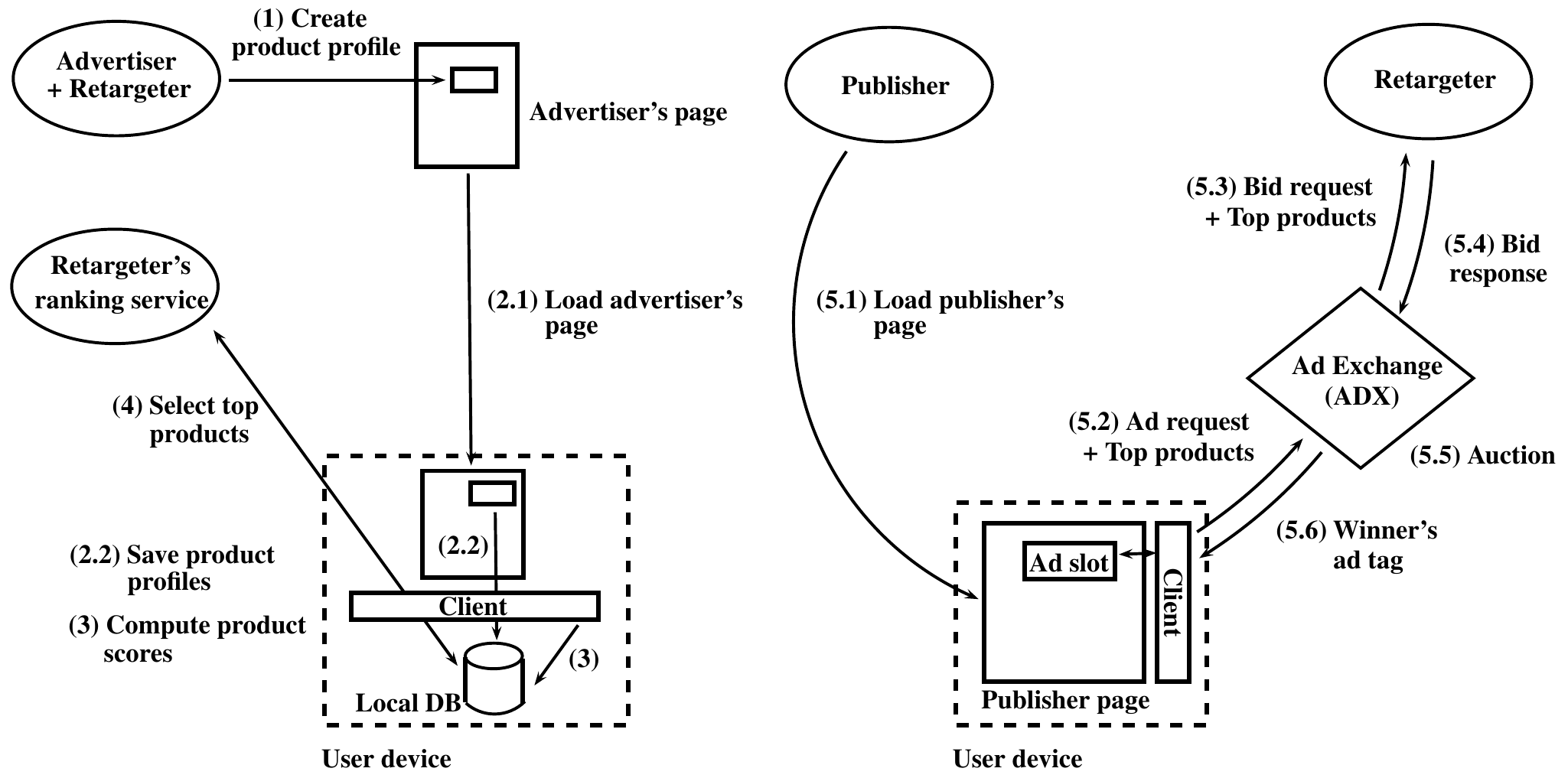}
  \caption{System overview}
  \label{fig:model}
\end{figure}

We follow a distributed (in contrast to the currently centralized) approach where a software agent running at the user's device, called \textit{client}, creates and maintains user profile, as well as computes a \emph{score} for each visited product. These scores allow retargeters to select ads and to adjust their advertising prices to the user. At a high level, the protocol works as follows (Figure \ref{fig:model})\footnote{Note that most values that are processed in the protocol are encrypted. We ignore this aspect at this stage for simplicity.}: 

\begin{enumerate}
\item The retargeting advertiser builds product feeds for the retargeter in which it specifies the products to be advertised. The advertiser and the retargeter agree on a range of advertising details for each product such as product quality (e.g., inferred from the number of users showing interest in the product), user targeting criteria and ad pricing\footnote{These details can be set by advertisers or suggested by retargeters.}. These details are encoded for each product as a \textit{product profile} and embedded into the advertiser's web pages.

\item When a user visits a commercial web page (which belongs to a retargeting advertiser, e.g. hotels.com) looking for a product, the client retrieves the product profile from the page. 

\item The client then computes the score of this product by matching the product profile with the user profile. This score gives an estimation of the expected revenue of the retargeter from advertising this product to the user, and therefore shows whether this product is a good candidate for retargeting or not. 

\item The client selects the products having the highest scores, called \textit{top-products}, for each retargeter that it encountered.

\item When the user visits a publisher's page (e.g., accuweather.com), which contains an advertising frame, he sends the top $m$ (say $m = 3$) product ids and scores of each retargeter to the ADX\footnote{We assume a mechanism such that retargeters notify users of which ADXs they are currently working with (e.g. by putting this information on a website accessible to clients). The client only includes products of the related retargeters when it detects an ADX.}. The ADX initiates the RTB auction among all retargeters by distributing the top-products of each retargeter. Based on the scores of the user's top products, each retargeter decides whether to serve an ad to the user and at what price, and sends the bid to the ADX. Finally, the ADX puts the winner's ad creative on the publisher's web page; the ad creative anonymously loads ads from the retargeter through a proxy mechanism at the ADX.  

\end{enumerate}
As product profiles are considered as commercially sensitive information, they are encrypted by retargeters using a homomorphic encryption scheme such that users cannot access the profile attributes. In order to select the highest scored products, the client invokes the \emph{ranking service} of each retargeter. In particular, the client sends the list of \textit{encrypted} scores to the ranking service in a way (described later) that leaks neither the list of products nor the user profile to the retargeter. The ranking service decrypts the scores, sorts them, and sends back the sorted list of encrypted scores to the user.

The ADX \textit{does not} include any user cookie (or user identifying information) into bid requests as in today's RTB protocol. Meanwhile, the top products sent to retargeters are encrypted and therefore inaccessible to the ADX. Retargeting ads are loaded through a proxy mechanism at the ADX and similarly are encrypted in order to prevent the ADX from learning their content. Note that our scheme requires that retargeters buy ad spaces from publishers only through ADX\footnote{We believe this requirement is reasonable: major retargeters (e.g. AdRoll, Criteo) are actually partners with major ad exchanges (e.g. DoubleClick, RightMedia, Facebooke, etc.) for their \textit{indirect buying} through RTB \cite{adrollpartner}\cite{criteodoubleclick}\cite{criteofacebook}, and start embracing Preferred Deals as an efficient technique for their \textit{direct buying} relationships with publishers \cite{directdeals}\cite{preferreddeals}\cite{adform}\cite{wpd}.}.
\section{System Details}
\label{sec:details}
\subsection{Product Score Evaluation}
\label{sec:prod_score}

We present our approach to compute product scores in the Cost-Per-Click (CPC) model, i.e., advertisers only pay retargeters when users click on ads. A similar approach can be applied to other models such as Cost-Per-Mile or Cost-Per-Action. In CPC, the retargeter's expected revenue from advertising a product $P$ to a user $U$ is calculated as $CTR_U(P)\times CPC(P)$, where $CTR_U(P)$ is the estimated Click-Through-Rate (CTR) of the ad shown to this user, and $CPC(P)$ is the price that the advertiser pays for an ad click. The product score in our scheme is defined as this expected revenue.

We assume that each product $P$ has a default $CTR$, denoted by $CTR(P)$, and computed by a retargeter $R$, e.g., based on the click history of similar products. A product initial score (PIS) of $P$ is calculated as $PIS = CTR(P)\times CPC(P)$, which is independent of users. $R$ targets users based on a set of \textit{user attributes} such as $\{\mathit{gender, age, interests, location}\}$. Example values of these attributes are $\{``male", ``24-35", ``sport", ``Paris"\}$. For each of these attribute values, such as $male$, $R$ quantifies its effect on the $CTR$ of $P$ by an \textit{impact factor}, such as $1.2$, indicating that advertising $P$ to a $\mathit{male}$ user would increase $P$'s $CTR$ by $20\%$. Impact factors can be learned from statistics, e.g., by analyzing the $CTR$ of similar products when being advertised to $\mathit{male}$. The impact factor which is larger/smaller than $1$ increases/decreases $P$'s $CTR$. 
If the retargeter does not have sufficient statistics to measure an impact factor, it sets that to a default value (e.g., $1$).

The retargeter $R$ configures, for each product, the initial score $PIS$ and the impact factors for all possible values of each user attribute. These properties are then encrypted by $R$'s symmetric key using a homomorphic encryption scheme such that a user can select the encrypted impact factors corresponding to his profile (e.g., a $male$ user picks the encrypted impact factor for $male$) without knowing their values. The user performs this selection and then computes the product score using the (encrypted) $PIS$ and the selected (encrypted) impact factors. In the following, we describe the details of our approach. 

\subsubsection{Encoding User and Product Profiles.}
The user attributes in our scheme include \textit{age}, \textit{gender}, \textit{location} and \textit{interest categories}. In addition, they can also include three extra attributes that are different for each product: \textit{user conversion status} (e.g., the product was put into shopping cart but not purchased), \textit{frequency of visits} (e.g., visits per day) and \textit{time of last visit} (e.g., last hour or last day)\footnote{Users who put the product into shopping cart are more likely to make a purchase than those only looking at the product. Larger visit frequency usually anticipates a purchase. Similarly, recently visited products are more likely to be purchased.}. Note that this list is not exhaustive: additional attributes can be added depending on specific targeting purposes. 

\begin{itemize}
\item \textbf{User profile:} A user profile is described by a vector $U$ where each coordinate encodes the value of one user attribute. Specifically, each coordinate is the index of the value of the corresponding attribute in its value set. An example of a user profile is $U = (1, 0, 89, 1, 2, 2, 3)$ where each coordinate may encode 
$(1: ``18-24"(age), 0 : \mathit{``male"}(gender), 89 : \mathit{``Madrid"}(location), 1 : \mathit{``computer\ science"}(interest), 2: \mathit{``in\ shopping\ cart"}(conversion\ status), 2 : \mathit{``10\ visits/day"}(frequency), 3 : \mathit{``last\ week"}(last\ visit))$.
\item \textbf{Product profile:} A product profile belonging to a retargeter $R$ includes (1) the ids of $P$ and $R$ ($id_P$ and $id_R$, resp.); (2) $PIS$; (3) a url of $R$'s ranking service (see later); (4) a set of vectors $\{F_1, ..., F_n\}$, where each corresponds to a user attribute and contains impact factors for all possible values of that attribute. Given the previous example, $F_2$ corresponds to $gender$, while $F_2[1]$ and $F_2[2]$ are the impact factors of $male$ and $female$, respectively. \footnote{Note that all values of the impact factors are converted to integer by using the micro format  (e.g., $1$ is converted to $1,000,000$ micros). The micro format is commonly used in the advertising industry.}
\end{itemize}

\subsubsection{Computing Score.} 
For each attribute $i$, the client uses $U[i]$ to pick one impact factor from $F_i$, namely $F_i[U[i]]$. The score $\mathit{S_P^R}$ for product $P$, which belongs to a retargeter $R$, is computed as follows:

\begin{equation}
\label{eq:score}
S_P^R= \mathit{PIS_P^R} \times \prod\limits_{i=1}^{n} F_i[U[i]]
\end{equation}

Although this formula can be computed using a \textit{multiplicative} homomorphic encryption scheme\footnote{An encryption scheme, denote by $\mathit{Enc}$, is additive homomorphic if, given two arbitrary plaintexts $m_1$ and $m_2$, it allows computing $\mathit{Enc}(m_1 + m_2)$ from $\mathit{Enc}(m_1)$ and $\mathit{Enc}(m_2)$ without decrypting any of these values. Similarly, the scheme is multiplicative homomorphic if it allows computing $\mathit{Enc}(m_1*m_2)$ from the ciphertexts $\mathit{Enc}(m_1)$ and $\mathit{Enc}(m_2)$.}, such as El-Gammal \cite{ElGammal}, such asymmetric encryption is too costly in practice. 
We, instead, use an \textit{additive} homomorphic encryption scheme based on \textit{symmetric keys}, which is proposed in \cite{ccastel} (Appendix), for its efficiency. Hence, we convert Formula \ref{eq:score} to additive form by taking the logarithm of both sides and use the resulting formula with the additive scheme:

\begin{equation}
\label{eq:score_add}
\log S_P^R= \log \mathit{PIS_P^R} + \sum\limits_{i=1}^{n} \log F_i[U[i]]
\end{equation}

For all $i$, $R$ encrypts the logarithm of $F_i$ such that each coordinate of $F_i$ is encrypted with a different key.  Specifically, for any $F_i$, $R$ computes 
$F_i^E = (\mathit{Enc}_{k_{i,1}}(\log F_i[1]), \ldots, \mathit{Enc}_{k_{i,|F_i|}}(\log F_i[|F_i|]))$ 
where $k_{i,j} = \mathit{Hash}(id_P|K|i|j)$ and $K$ is the secret key of $R$.  In addition, $R$ computes $\mathit{Enc}_{k_{\mathit{PIS}}}(\log \mathit{PIS}_P^R)$, where $k_{\mathit{PIS}} = \mathit{Hash}(id_P|K|``PIS")$.

To compute the product score, the user simply applies Formula \ref{eq:score_add} but in the encrypted domain, i.e.,
\begin{equation}
\label{eq:score_enc}
\mathit{Enc}_{k_{\mathit{PIS}}}(\log \mathit{PIS}_P^R) + \sum_{i=1}^n F_i^E[U[i]] = \mathit{Enc}_{k_{\mathit{PIS}} + \sum_{i=1}^n k_{i,U[i]}}(\log S _P^R)
\end{equation}
where the equality follows from the homomorphic property of $\mathit{Enc}$ \cite{ccastel}.
The product score can be retrieved by taking the exponent of the decrypted value. 

\subsection{Product Ranking}
The client needs to rank products belonging to the related retargeter $R$ when it encounters a new product profile. Alternatively, it can perform the ranking periodically (e.g., hourly) in case new products are frequently recorded. Recall that the URL of the ranking service is included in each product profile of $R$.

To rank a set of products, the client sends the list of product scores, computed in Formula \ref{eq:score_enc}, to the ranking service of $R$. However, in order to decrypt this score, $R$ would need the set of keys $\sum_{i}k_{i,U[i]}$, which eventually reveals the user profile $U$. We instead follow a popular approach \cite{obliviad}\cite{singleround} and propose to implement the whole ranking procedure using a secure co-processor (SC) (e.g., IBM 4765 \cite{ibm4765}), which provides secure storage as well as trustworthy, programmable execution environment. The SC could be deployed at $R$, and users can verify the code being executed on the SC through a remote code attestation mechanism (e.g., \cite{scauthenticate}). As the SC is tamper-resistant, while the communication between the client and the SC is encrypted, no other parties (including $R$) will learn anything about user profiles and top-products.

The ranking algorithm is simple and can be public. In particular, $R$ installs $K$ and the public ranking procedure on the SC. Afterwards, the client establishes a secure connection to the SC (e.g., through TLS), and sends $U$, as well as $id_P$ with $Enc_{k_{\mathit{PIS}} + \sum_{i}k_{i,U[i]}}(\log S _P^R)$ for the related products to the SC. Then, the SC can decrypt product scores, rank these products, and send back the sorted list of product ids to the client. 

\subsection{Ad Serving}
The client sends ad requests, which include the top products' $ids$, $id_Rs$ (retargeter ids), $PIS$ and scores, to the ADX. Note that, as opposed to what in current RTB systems, ad requests in our scheme do not contain any cookie. The ADX uses $id_Rs$ to separate the top-products for each retargeter and includes them along with the user's visiting page in related bid requests sent to appropriate retargeters. Each retargeter $R$ decrypts the product scores and adjusts these scores according to the quality and relevance of the visiting page. If there is a difference between the $PIS$ of a product and the latest $PIS$ at the retargeter (e.g. due to a CPC change), it updates the product score accordingly. $R$ then selects the product with highest score and determine the bid price based on the score. Subsequently, $R$ builds an ad creative which contains the ad url to load ad for this product from its server, and then includes the creative and the bid price in a bid response which is sent to the ADX.

If $R$ wins the auction, its ad creative is sent to the user's device. Note that current ADXs mostly allow these creatives to load ads directly from the retargeter's server or from the ADX' storage which contains retargeters' pre-uploaded ads. Both approaches do not protect user privacy: while the former exposes user IP address to the retargeter, the latter reveals personalized ad content to the ADX. In our scheme, we protect user privacy by requiring that ads be encrypted by the retargeter and loaded through a proxy mechanism at the ADX. The resulting computational and bandwidth overhead is analyzed in Section \ref{sec:exp}.
 
In particular, the ADX replaces the ad url in $R$'s ad creative with a url pointing to the ADX which contains the original url as value of a HTTP parameter. At the user's device, the creative requests the ad from the ADX; the ADX retrieves the retargeter's ad url, loads the ad from the retargeter, and then returns the ad to the user. The ad view or click report is sent to the ADX, which subsequently removes any user related information (e.g. IP address, browser info, etc.) and forwards the report to the retargeter.

\paragraph{Protecting Ad Requests and Content:} The product ids (in ad requests) or the ad content could reveal the user's top products to the ADX. In order to protect these data, we use a symmetric key which is shared between the retargeter and the client. In particular, before sending an ad request, the client generates a session key $K_{u,R}$  for each retargeter $R$ whose products are selected for the request. The client then encrypts each session key with the public key of the corresponding retargeter and includes the resulting encrypted key in the ad request. The ADX subsequently includes the encrypted keys in bid requests. $R$ first decrypts $K_{u,R}$ with its private key and then decrypts the top-products' ids with $K_{u,R}$. $K_{u,R}$ is also used by $R$ to encrypt the ad content for the client if $R$ wins the auction. 

\subsection{Other Features}
\paragraph{Frequency Capping:} Frequency capping restricts the number of times an ad is shown to the same user (e.g., less than 10 times per day). This is solved trivially in our scheme: the client counts the number of times a product is advertised to the user, and ignores a product when building ad requests if its counter is beyond a threshold. 
\paragraph{Click Fraud Defense:} Since the click report is anonymized (by the proxy mechanism at the ADX), 
this makes the click fraud more difficult to detect. However, a similar approach to \cite{privad} could be applied: the retargeter sends suspected click reports to the ADX; the ADX traces back to the user IPs which are responsible to these clicks to examine the probability of click fraud. 
\section{Privacy Analysis}
\label{sec:analysis}

In this section, we analyze how user privacy is protected from retargeters, ad exchanges, advertisers, and other users. We also analyze how the confidentiality of the retargeter's ad selection algorithm is guaranteed.

\subsection{Retargeter}
A retargeter $R$ might get user information from bid requests or ranking requests. Although $R$ gets top-products from bid requests, it cannot associate them with user identifying information such as the IP address (which is not included in bid requests). Contrarily, while $R$ can see the user IP from ranking requests, it is not able to get any related user data since these requests are received and processed by the SC through a secure connection. The tamper-proof property of SC prevents $R$ from intercepting and learning any internal data during the ranking process. In summary, $R$ \emph{unlinkably} gets users' top products (from bid requests) and users' IP addresses (from ranking requests). 

Without covert channel or collusion, $R$ cannot break \textit{user anonymity} unless it can correlate ranking requests with bid requests, e.g., if there are too few users, or ranking request time correlates with bid request time. Although both attacks seem impractical with large number of users (which is likely the case), the time correlation attack can be further mitigated by randomizing the time of ranking requests.

Nevertheless, $R$ might attempt to link users' top products from different bid requests and gradually build unique profiles of users (\textit{linkability}). This would be difficult given that the client sends only a small number of top-products in ad requests. A possible attack is to infer user attribute values from the product scores and then use them as a fingerprint to link different top products. This, however, can be mitigated by coarsening score values at ranking (performed inside the SC) and then using these coarsened values in ad requests. 

Although $R$ can hardly break user anonymity or profile unlinkability if it faithfully follows the protocol, it might collude with other parties in order to do that. Recall that we assume non-collusion between ADX and $R$ (Section \ref{sec:goal}). Other parties that might (and in fact have motivation to) collude with $R$ are the advertiser and the publisher.  


Assume that $R$ colludes with a malicious advertiser. For example, the advertiser might send all its log entry database (possibly containing visited products, time of visits, and  the user IP address of every visit) to $R$. In order to link bid requests, which likely contain products from other advertisers, with IP addresses from this database, $R$ may perform timing or product frequency analysis.

\begin{itemize}
\item \textit{Timing analysis:} $R$ correlates the visiting time related to a product (in the malicious advertiser's database) with the reception time of a bid request (e.g., if they are close, both events are possibly originated from the same IP).
However, as discussed previously, time-based correlation is difficult due to the large number of users. In addition, it is hard to predict the interval between a visit to an advertiser's page and a visit to a publisher's page. 

\item \textit{Product frequency analysis:} $R$ selects products from the database which were visited by a small number of users, ideally only by a single user (identified by IP address). If a bid request contains any of these products, $R$ can associate it with a small group of users, or a single user, accordingly. Though this attack is possible, it can only affect a very small proportion of users. In addition, in this form of attack, $R$ cannot actively target a victim.

\end{itemize}

The timing attack would be more practical if $R$ colludes with a malicious publisher, as the time of a visit to a publisher page is very close to the time of the resulting bid request. However, this only works in case of very small publishers which have few connecting users at a time. In addition, the client can add some arbitrary delays in sending bid requests in order to mitigate the risk of such attack.

\subsection{Ad Exchange}
Although the ADX can see the user's IP address, it cannot obtain the list of top-products or ad content; they are encrypted using the session key shared between the user and the retargeter. The ADX cannot forge fake session keys in order to decrypt retargeting ads served through its proxy mechanism. In particular, if a fake session key is produced, the retargeter cannot get the right product ids from the related request.

Notice that, as part of the RTB protocol, ADXs can still get urls of visited sites in our solution. Nevertheless, the risk of profiling users using these urls is less severe than that in retargeting (Section \ref{sec:goal}). Moreover, users can block cookies to prevent possible tracking performed by the ADX; our scheme works without the need of any tracking cookie. In this work, we focus on the privacy risks of retargeting, and leave the total elimination of this url leakage towards the ADX for our future work. 

\subsection{Advertiser}
In our system, an ad click brings the user directly to the advertiser's site, the same as what is happening today. We acknowledge that this is a problem since advertisers can leverage fine-grained targeting feature and high user profiling quality in our scheme to learn more information about users than what they can learn in today's system. One solution could be to handle post-click sessions through an anonymized network, such as TOR \cite{tor}. This might, however, lead to additional network latency, complexity in measuring ad performance or other implication, and therefore need deeper analysis from the research community and the advertising industry. 

\subsection{User}
Targeted ads were proved to be a potential source of leaking user private information due to its personalized content \cite{adsleak}. In case of retargeting ads, someone happens to look at a user's screen when he is browsing the web in a public place (e.g. at work) may infer his previous private actions (e.g. at home)\footnote{We assume that the user is using the same laptop computer in both environments.} that the user may want to keep secret (e.g. looking for an engagement ring). Since user profiles in our scheme are built and stored locally, users can trivially filter out sensitive products to prevent such unexpected information leakage. The encryption of transmitted data also prevents eavesdroppers from exploiting personalized ad content thereby inferring user private information.

The client can be malicious toward a retargeter $R$. For example, one of $R$'s competitors might install and manipulate a client in order to learn product profiles of $R$. Since the used encryption scheme is proved to be \textit{perfectly secure} \cite{ccastel} (Appendix), only $R$ can decrypt its product profiles. Nevertheless, a malicious client may attempt to manipulate its own user profile to observe possible changes in the ranking list and thereby inferring some properties of the product profiles. This kind of attack needs to be repeated many times in order to learn meaningful results. Consequently, it can be mitigated by applying a threshold on the number of ranking requests per client. 
Furthermore, the ranking service on the SC may apply a randomization in the order of top-products to make such kind of attacks more difficult, if not impossible.   
\section{Implementation and Evaluation}
\label{sec:exp}

\subsection{Implementation}
\label{sec:implement}
To build the client, we extend the Firefox plugin \textit{HttpFox} \cite{httpfox} and use SQLite for storing local data. The Ad Exchange and the Retargeter are written using NodeJS \cite{nodejs}. Note that the ranking service, which is supposed to be implemented on a SC, is implemented in NodeJS and executed by a normal processor in our implementation. We ran our client inside a Firefox browser on a laptop running OS X 10.7.2 on an Intel Core i5 2.4 GHz, and the retargeter and ad exchange on a machine with an Intel Core 2 Duo 2.66 GHz and running Ubuntu 11.04.

Based on Google Ads Settings \cite{googleadssettings}, we configure the system with 2 genders, 7 age ranges, 24 top interests and 846 word localities. The user conversion status, the frequency of visits, and the time of last visit are all configured with 5 permissive values. The client stores up to 1000 products (for all retargeters) and includes 3 products per retargeter in each ad request. 

\subsection{Evaluation}
We evaluate, in this section, the computational and bandwidth overhead of our scheme (compared to the existing system) through an example scenario.

\paragraph{Example scenario:}
A retargeter $R$ provides retargeting services for 100 advertisers, each having 1000 retargeting products. Each day there are 1 million unique users browsing these advertisers' websites (10K users per site on average). Each user browses for 20 retargeting products which belong to 3 retargeters on average, and issues 20 ad requests containing such products, per day. Ad requests are handled by a single Ad Exchange (ADX), resulting in 20 millions requests per day (about 232 requests per second) received by the ADX. For each ad request, the ADX sends bid requests to all retargeters whose products are contained in the ad request. The retargeter R submits bid responses for all 20 millions bid requests, wins about 10\% of these auctions (2 millions winning times). 


\subsubsection{Computational Overhead.} 
\label{sec:computation}
Our client can perform 5K homomorphic computations per second. It can generate 200 session keys (encrypted by $R$'s public key) and 30 ad encryptions per second. With a few dozen score computations and ad requests per day, the computation at client causes a negligible overhead.

Our retargeter can perform 133 product encryptions, 100K score decryptions, and 6K ad encryptions per second. The computational overhead of the retargeter (per day) is shown in Table \ref{retargeter_com_overhead}. In this table, (1)(the first line) is estimated in the worse case (the retargeter re-encrypts all 100K products everyday), while (3) can be significantly optimized by offloading asymmetric operations using dedicated hardware \cite{sslshader}. With our implementation, assuming that the retargeter rents computation resources from Amazon EC2 \cite{ec2} (e.g., using a c3.large instance which is optimized for computation purposes), the total computational overhead costs about \$0.45 per day.

\begin{table}[!t]
\renewcommand{\arraystretch}{1.3}
\caption{Computational overhead at retargeter (per day)}
\label{retargeter_com_overhead}
\centering
\begin{tabular}{|c|l|r|}
\hline
Index & Action & Time (hours)\\
\hline
1 & Encrypting 100K products on advertisers' websites & 0.21\\
2 & Decrypting 60M scores in bid requests & 0.83\\
3 & Decrypting 2M session keys & 1.23\\
4 & Encrypting 2M ads with session keys & 0.46\\
5 & Processing 20M ranking requests & 0.33\\
\hline
& Total & 3.06\\
\hline
\end{tabular}
\end{table}

\subsubsection{Bandwidth and Storage Overhead.}
\label{sec:bandwidth}
Each product profile increases the size of the product web page by 6 KB. This is a negligible overhead given the fact that, for example, a maty.com's product page's html source is around 100 KB, excluding images, css and javascript files. The sizes of an ad request and a ranking request are 23.94 KB and 15.96 KB, respectively. The total bandwidth overhead for each user is therefore approximately 2 MB per day. The products stored at a user's device (maximum 1K) cause a maximum 8MB local storage. 

\begin{table*}[!t]
\renewcommand{\arraystretch}{1.3}
\caption{Bandwidth overhead at ADX and retargeter (per day)}
\label{table_bwth_overhead}
\centering
\begin{tabular}{|c|l|c|l|c|}
\hline
 &\multicolumn{2}{|c}{\textbf{ADX}} &
\multicolumn{2}{|c|}{\textbf{Retargeter}} \\ 
\hline
Index & Action & Bwth & Action & Bwth \\
\hline
1 & Receiving 20M ad requests & 53GB & Receiving 20M bid requests & 18GB\\
2 & Sending 60M bid requests & 53GB & Receiving 20M ranking requests & 35GB\\
3 & Proxying 10M retargeting ads & 190GB &&\\
\hline
\end{tabular}
\end{table*}

The bandwidth overhead at the ADX and the retargeter are presented in Table \ref{table_bwth_overhead}. In order to provide an estimation of the resulting cost, we assume that both the retargeter and the ADX run their software on Amazon EC2's servers. Amazon EC2 only charges the bandwidth from EC2 to the Internet, which is related to ads served by the ADX to users. If we apply the upper bound of this price, namely \$0.12 per GB, the bandwidth overhead resulting from serving ads (190 GB) costs the ADX approximately \$22.8 per day (\$$2.28*10^{-6}$ per ad). This additional cost can be shared among retargeters, advertisers, and publishers, for example by the ADX slightly increasing the transaction commission. Given the significant number of advertisers and publishers working with an ADX in general, the cost per entity would become negligible.

\section{Discussion}
\label{sec:discuss}


\subsection{Compatibility}
Our purpose is not to replace but rather complement the current retargeting system by providing an alternative solution for retargeters (and ad exchanges) to provide retargeting ad service in a privacy-preserving manner. An example scenario would be as follows. The retargeters choose privacy-preserving (our scheme) or regular (current scheme) RTB mode at the ADX. At auction, the ADX sends privacy-preserving retargeting bid requests (as described in our scheme) to privacy-advocate retargeters, and regular bid requests to the others. The ADXs which support our scheme specify this in their ad requests (e.g., by using a special HTTP header, such as $``PPRetargeting=true"$). The client\footnote{Privacy-advocate retargeters and ADXs encourage users to use the client software.} intercepts and includes into these requests the respective user's top retargeting products. If users do not favor traditional retargeting or RTB, they can, for example, configure the web browser to disable third-party cookies.

\subsection{Scoring Algorithm}
The score computation (i.e., $CPC(P)\times CTR(P)\times \prod\limits_{i=1}^nF_i$), is based on an assumption that the effect of user attributes on the product's $CTR$ are independent of each other. Consequently, this algorithm has a limitation: it does not take into account the intrinsic correlation between user attributes. For example, say a user's age ``18-24" and interest ``sport", each increasing a product's $CTR$ with 1.2 (20\%) and 1.1 (10\%), respectively, a combination of them might not necessarily be equal to $1.2 \times 1.1 = 1.32$, but can be higher or lower than that depending on the correlation between the two attributes. In the following, we discuss a possible extension that takes into account this correlation.

\paragraph{Attribute Coefficients:} We assume that the correlation between two attributes $F_i$ (e.g., age) and $F_j$ (e.g., gender) can be quantified by a coefficient $C_{ij}$, so that their combined effect on a product's $CTR$ is computed as $F_i\times F_j \times C_{ij}$. Each $C_{ij}$ can be defined with different values for different value ranges of $F_i$ and $F_j$. The score in this case would be computed as: $CPC(P)\times CTR(P)\times \prod\limits_{i=1}^nF_i\times \prod\limits_{i,j}C_{ij}$. Note that the correlation coefficient can be computed for more than two attributes, e.g., $C_{ijk}$ or $C_{ijkl}$. In the worse case, the size of the coefficient set is equal to the number of all possible combinations of attribute values. The coefficient values are also encrypted by the retargeter, and their applied attribute ranges can be defined in a script encoded into the product profile. 

\subsection{Gathering Statistics}
User statistics (e.g., click behaviors of a group of users) are important for retargeters to enhance their retargeting performance. In our scheme, the SC can also be used to aggregate this information in a privacy-preserving manner. In particular, after an interval (e.g., a week), the client sends the user's profile and CTRs of local products to the SC (through a secure connection). The SC aggregates these data from a sufficient number of users, produces statistical results (e.g., average CTR of a product advertised to \textit{sport-enthusiastic} users), and sends them to the retargeter.

\section{Related Work}
\label{sec:related}
There have been several research proposals aimed at designing a privacy-preserving targeted advertising system. Nevertheless, none of them address privacy problems in retargeting or consider RTB in their design. In the following, we survey these proposals and analyze their differences from ours.

Saikat et al. proposed Privad \cite{privad}. In their design, a client software builds the user profile locally, a \textit{dealer} acts as an anonymizing proxy to mediate the communication between users and ad brokers, while the communication is protected from the dealer with the public keys of ad brokers. Privad uses a publish-subscribe mechanism for ad delivery: the ad broker transmits ads to users according to their subscribed coarse-grained information (e.g., generic interest category such as ``sport"), the ads are cached at the user's device, and the local software selects ads that best match the user profile when encountering an ad box. Since the client in Privad is implemented as an untrusted black box, a reference monitor is needed to gauge the traffic between the client and the network. 

The publish-subscribe and ad caching mechanism in Privad is not appropriate to real-time auctions in RTB, in which ads and bid prices are dynamically selected at real-time. Our approach is somewhat similar to Privad in terms of using a proxy to protect user anonymity. However, while Privad requires an additional party, the dealer, whose incentive and business model are not clear, we leverage an existing entity, the ad exchange, and its RTB protocol, which plays an important role in the current ad system. 

Moreover, Privad does not detail how to protect the ad selection algorithm of ad brokers from users. In addition, as the client is untrusted, it would be difficult (e.g., might require significant human intervention) to prevent the client from bypassing the reference monitor to leak user information to the network (e.g., through semantic means). Furthermore, in Privad, as a client software is provided by an ad broker, each user might have to install as many clients as the number of ad brokers. This requirement imposes an important barrier for the entrance of new ad brokers. In our proposal, users trust the client (which can be open-source) and therefore do not need a monitoring mechanism. We use homomorphic encryption to allow the client to perform part of the retargeter's ad selection algorithm while protecting its confidentiality. Each user installs only one software that can work with any retargeter which follows the protocol. 

Adnostic \cite{adnostic} similarly profiles users locally but uses a different approach for ad rendering. In Adnostic, the broker transmits a set of ads solely based on the page that a user is currently visiting, and the local software selects ads that best match the user profile to display to the user. Nevertheless, since only a small number of ads are transmitted by ad brokers without any knowledge about the user, the targeting performance in Adnostic is far from optimal. In addition, as the ad broker cannot see which ad is displayed to the user, it is difficult to control frequency capping. Finally, this model is not suitable to RTB which allows only one ad to be sent by a bidder in each auction. In our scheme, we ensure fine-grained targeting, the frequency capping is solved trivially, and RTB is supported. 

Obliviad \cite{obliviad} aims to shift the current algorithms of ad brokers into secure co-processors (SC). The user sends his profile in form of keywords to the SC through a secure connection (e.g., TLS). The SC subsequently selects ads from the broker's database using these keywords. Obliviad implements a PIR scheme over an ORAM structure to hide all the database access patterns from the broker. One similar characteristic between Obliviad and our proposal is the use of SC as a trusted environment on the advertiser side. However, while Obliviad aims to implement the whole broker's ad selection algorithm inside the SC, we only leverage the SC for a simple sorting algorithm. Note that, due to tamper-resistance requirements, the SC is often significantly constrained in both computation ability and memory capacity in comparison with host CPUs \cite{upir}. Our sorting algorithm is extremely simple (i.e. decrypt and sort numbers), and therefore is easy to be implemented inside the SC. In addition, similarly to Privad and Adnostic, it is unclear how to adapt Obliviad to RTB. For example, it is apparently impractical to make a TLS connection to the SC through a RTB auction.

RePriv \cite{repriv} proposes a client-side framework for content providers to inject their \textit{miners} (in form of embedded codes) to run on the user's terminal. These miners collect user information according to specific purposes of each provider, then customize the provided content accordingly. This approach can be applied in our scheme, for example in building and maintaining user profiles.
\section{Conclusion}
\label{sec:conclusion}
Retargeting ads are growingly rampant and cause great privacy concerns, mostly resulting from tracking user intents. In this work, we propose the first retargeting system that does not rely on tracking. Our scheme leverages homomorphic encryption to distribute the ad selection algorithm between the user and the retargeter in a way that securely combines confidential data from both parties. The proposed scheme is compatible with RTB and supports major characteristics of current ad systems such as real-time auction, fine-grained targeting, ads freshness and frequency capping. 

This is, to our knowledge, the first work that considers RTB in a privacy-preserving advertising solution. We note that RTB is increasingly prevalent and plays an important role in today's targeted advertising systems, itself exposing serious privacy concerns \cite{rtb}. Enhancing user privacy in RTB in general is therefore our goal in a near future. 


\begin{thebibliography}{4}

\bibitem{nytimes_retargeting} The Wall Street Journal: Retargeting Ads Follow Surfers to Other Sites, \url{http://www.nytimes.com/2010/08/30/technology/30adstalk.html} (2011)

\bibitem{dnt} Do Not Track, \url{http://donottrack.us/}
\bibitem{ghostery} Ghostery, \url{http://www.ghostery.com/}
\bibitem{adblockplus} Adblock Plus, \url{https://adblockplus.org/}
\bibitem{trackmenot} TrackMeNot, \url{http://cs.nyu.edu/trackmenot/}
\bibitem{dntme} DoNotTrackMe, \url{https://www.abine.com/dntdetail.php}
%
\bibitem{ec2} Amazon Elastic Compute Cloud (Amazon EC2), \url{http://aws.amazon.com/ec2/}

\bibitem{privad} Guha, S., Cheng, B., Francis, P.: Privad: Practical Privacy in Online Advertising. NSDI, Boston, MA (2011)

\bibitem{sslshader} Jang, K., Han, S., Han, S., Moon, S., Park, K.: SSLShader: Cheap SSL Acceleration with Commodity
Processors. NSDI, Boston, MA (2011)

\bibitem{adnostic} Toubiana, V., Narayanan, A., Boneh, D., Nissenbaum, H., Barocas, S.: Adnostic: Privacy Preserving Targeted Advertising. NDSS, San Diego, CA (2010)

\bibitem{obliviad} Backes, M., Kate, A., Maffei, M., Pecina, K.: ObliviAd: Provably Secure and Practical Online Behavioral Advertising. S\&P, San Francisco, CA (2012)

Matthew Fredrikson and Ben Livshits

\bibitem{repriv} Fredrikson, M., Livshits, B.: RePriv: Re-Envisioning In-Browser Privacy. S\&P, Oakland, CA (2011)

\bibitem{scauthenticate} W. Smith, S.: Outbound authentication for programmable secure coprocessors. ESORICS, Zurich, Switzerland (2002)

\bibitem{singleround} Williams, P., Sion, R.: Single round access privacy on outsourced storage. CCS, New York, NY (2012)

\bibitem{ccastel} Castelluccia, C., Mykletun, E., Tsudik, G.: Efficient Aggregation of encrypted data in Wireless Sensor Networks. In: Second Annual International Conference on Mobile and Ubiquitous systems: networks and services. San Diego, CA, USA (2005).


\bibitem{criteodoubleclick} Criteo gains great results and scale by 
retargeting audiences through real-time 
bidding with DoubleClick Ad Exchange, \url{http://doubleclickadvertisers.blogspot.fr/2011/06/criteo-gets-great-results-retargeting.html} (2011)
\bibitem{criteofacebook} Criteo to Provide Customers With Access to Facebook Exchange, \url{http://www.criteo.com/en/news-and-events/press-releases/criteo-provide-customers-access-facebook-exchange} (2012)
\bibitem{adrollpartner} AdRoll's ad exchange partners, \url{http://www.adroll.com/about/partners}

\bibitem{directdeals} Introducing Ad Exchange Direct Deals, \url{http://doubleclickpublishers.blogspot.fr/2011/09/introducing-ad-exchange-direct-deals.html} (2011)
\bibitem{preferreddeals} Preferred Deals: A New Way to Sell in the Ad Exchange, \url{http://doubleclickpublishers.blogspot.fr/2012/06/preferred-deals-new-way-to-sell-in-ad.html} (2012)
\bibitem{adform} Preferred Deals | fixed priced deals made easy, \url{http://blog.adform.com/real-time-bidding/preferred-deals-fixed-priced-deals-made-easy/} (2013)
\bibitem{wpd} Extra! Extra! Washington Post Digital Goes Programmatic, Gets Premium Rates with DoubleClick’s Ad Exchange, \url{http://www.google.com/think/case-studies/wpd-adx.html} (2013)

\bibitem{googleadssettings} Google Ads Settings, \url{https://www.google.com/settings/u/0/ads}

\bibitem{ibm4765} IBM 4765, \url{http://www-03.ibm.com/security/cryptocards/pciecc/support.shtml}

\bibitem{adsleak} Castelluccia, C., Kaafar, M.-A., Tran, M.-D.: Betrayed by Your Ads!. In: 12th international conference on Privacy Enhancing Technologies. Vigo, Spain (2012)



\bibitem{criteo6x} Targeting \& Retargeting Interview with Criteo, \url{http://behavioraltargeting.biz/targeting-retargeting-interview-with-criteo/} (2010)


\bibitem{shopping_habits} The New York Times: How Companies Learn Your Secrets, \url{http://www.nytimes.com/2012/02/19/magazine/shopping-habits.html} (2012)

\bibitem{httpfox} Martin Theimer: HttpFox Addon, \url{https://addons.mozilla.org/en-US/firefox/addon/httpfox/}
\bibitem{cm} Google: Cookie Matching, \url{https://developers.google.com/ad-exchange/rtb/cookie-guide}

\bibitem{tor} Tor Project, \url{https://www.torproject.org/}

\bibitem{nodejs} NodeJS, \url{http://nodejs.org/}

\bibitem{ElGammal} Smith, T.F., Waterman, M.S.: A Public-Key Cryptosystem and a Signature Scheme Based on Discrete Logarithms. IEEE Transactions on Information Theory 31, 469--472 (1985)

\bibitem{pricediscrimination} Price Discrimination is All Around You, \url{http://33bits.org/2011/06/02/price-discrimination-is-all-around-you/}

\bibitem{upir} Williams, P., Sion, R.: Usable PIR. NDSS, San Diego, CA (2008)

\bibitem{rtb} Olejnik, L., Tran, M-D., Castelluccia, C.: Selling Off Privacy at Auction. NDSS, San Diego, CA (2014)


%


%
%
%
%
%
%

\end{thebibliography}

\section*{Appendix: Additive Homomorphic Encryption Scheme from \cite{ccastel}}
\label{enc}
\subsection*{Description}
The main idea of the scheme is to replace the xor (Exclusive-OR) operation typically found
in stream ciphers with modular addition (+).

\begin{table}
\label{tbl:scheme1}
\centering
\begin{tabular}{l}
\hline
  Additively Homomorphic Encryption Scheme  \\
\hline 
\\
\parbox{9cm}{\leavevmode%
	Encryption:
	\begin{enumerate}
		\item Represent message $m$ as integer $m \in [0, M - 1]$ where $M$ is a large integer
		\item Let $k$ be a randomly keystream, where $k \in [0, M - 1]$
		\item Compute $c = Enc(m, k, M) = m + k\ (mod\ M)$.
	\end{enumerate}
	Decryption:
	\begin{enumerate}
		\item $Dec(c, k, M) = c - k\ (mod\ M)$
	\end{enumerate}
	Addition of Ciphertexts:
	\begin{enumerate}
		\item Let $c_1 = Enc(m_1, k_1, M)$ and $c_2 = Enc(m_2, k_2, M)$
		\item For $k = k_1 + k_2, Dec(c_1 + c_2, k, M) = m_1 + m_2$
	\end{enumerate}
}
\end{tabular} 
\end{table}

Assume that $0 \leq m < M$. Due to the commutative property of addition, the above scheme is additively homomorphic. In fact, if $c_1 = Enc(m_1, k_1, M)$ and $c_2 = Enc(m_2, k_2, M)$ then $c_1 + c_2 = Enc(m_1 + m_2, k_1 + k_2, M)$. 

Note that if $n$ different ciphers $c_i$ are added, then M must be larger than $\sum\limits_{i=1}^nm_i$, otherwise correctness is not provided. In fact if $\sum\limits_{i=1}^nm_i$ is larger than $M$, decryption will result in a value $m$' that is smaller than $M$. In practice, if $p = max(m_i)$ then $M$ should be selected as $M = 2^{\log_2 (p*n)}$.

The keystream $k$ can be generated by using a stream cipher, such as RC4, generated from a private key.

\subsection*{Security Analysis}
This additive homomorphic encryption scheme is very similar to a xor-based stream cipher and its security can be proven using a similar proof.

The security relies on two important features: (1) the keystream changes from one message to another and (2) all the operations are performed modulo an integer $M$. These two features protect the scheme from frequency analysis attacks. In fact, it can be proven that the scheme is \textit{perfectly secure}.

\newtheorem{propo}{Theorem}
\begin{propo} 
The previous encryption scheme is perfectly secure.
\end{propo}
\begin{proof}
For plaintext space $M$, keystream space $K$, let $\mathcal{K} = |M|$, $m \in [0; M - 1]$, $c \in [0; M - 1]$. Set $k^* = c - m\ (mod\ M)$. Then:
\begin{equation}
\begin{split}
\underset{k \leftarrow \mathcal{K}}{Prob}[Enc(k,m,M) = c] 
	&= \underset{k \leftarrow \mathcal{K}}{Prob}[k + m = c\ (mod\ M)] \\
	&= \underset{k \leftarrow \mathcal{K}}{Prob}[k = c - m\ (mod\ M)] \\
	&= \underset{k \leftarrow \mathcal{K}}{Prob}[k = k^*]
\end{split}
\end{equation}

If we assume that the maximum number of ciphertexts to be added is $n$ and that each plaintext is $l$-bit long, we must have $M = 2^{l+ \log(n)}$, i.e., $|M| = l + \log(n)$. If $c_i = (m_i + k_i)$, then the probability that $c_i \in [0, 2^l - 1]$ is twice the probability that $c_i \in [2^l;M - 1]$. More specifically, we have: $\underset{k \leftarrow \mathcal{K}}{Prob}[k = k^*] = 1/(2^l + M)$ if $c > 2^l$ and $\underset{k \leftarrow \mathcal{K}}{Prob}[k = k^*] = 2/(2^l + M)$ if $c < 2^l$.

Since these two equations hold for every $m \in \mathcal{M}$, it follows that for every $m_1,m_2 \in \mathcal{M}$ we have $\underset{k \leftarrow \mathcal{K}}{Prob}[Enc(k,m_1,M) = c] = \underset{k \leftarrow \mathcal{K}}{Prob}[Enc(k,m_2,M) = c]$ which establishes perfect security of the scheme.
\end{proof}

\end{document}